\documentclass[10pt,journal]{IEEEtran}

\ifCLASSOPTIONcompsoc
   \usepackage[nocompress]{cite}
\else
\fi
\ifCLASSINFOpdf
  \usepackage{graphicx,epstopdf}
   \graphicspath{{./images/}}
\else
\fi
\usepackage[cmex10]{amsmath}
\usepackage{amssymb,bm,cite,tikz,commath}
\usetikzlibrary{patterns,calc}
\newtheorem{theorem}{Theorem}
\newtheorem{lemma}{Lemma}

\DeclareMathOperator{\dgr}{deg}
\DeclareMathOperator{\tr}{tr}
\DeclareMathOperator{\Prob}{P}
\DeclareMathOperator{\E}{E}

\DeclareMathOperator{\Cov}{Cov}
\newcommand{\T}{^{\mathsf{T}}}
\DeclareMathOperator*{\argmin}{arg\,min}
\DeclareMathOperator*{\argmax}{arg\,max}

\begin{document}


%
\title{Accuracy of Range-Based Cooperative Localization in Wireless Sensor Networks:\\[2mm]\LARGE A Lower Bound Analysis}
%
%
%
%

\author{Liang~Heng,~\IEEEmembership{Member,~IEEE}
        and~Grace~Xingxin~Gao,~\IEEEmembership{Member,~IEEE}%
\thanks{The authors are with the Department of Aerospace Engineering and the Coordinated Science Laboratory, University of Illinois at Urbana-Champaign, Urbana, IL 61801, USA. E-mail: heng@illinois.edu, gracegao@illinois.edu.}}

\IEEEcompsoctitleabstractindextext{%
\begin{abstract}
Accurate location information is essential for many wireless sensor network (WSN) applications. A location-aware WSN generally includes two types of nodes: \emph{sensors} whose locations to be determined and \emph{anchors} whose locations are known a priori. For range-based localization, sensors' locations are deduced from anchor-to-sensor and sensor-to-sensor range measurements. Localization accuracy depends on the network parameters such as network connectivity and size. This paper provides a generalized theory that quantitatively characterizes such relation between network parameters and localization accuracy. We use the \emph{average degree} as a connectivity metric and use geometric \emph{dilution of precision} (DOP), equivalent to the Cram\'{e}r-Rao bound, to quantify localization accuracy. We prove a novel lower bound on expectation of average geometric DOP (LB-E-AGDOP) and derives a closed-form formula that relates LB-E-AGDOP to only three parameters: average anchor degree, average sensor degree, and number of sensor nodes. The formula shows that localization accuracy is approximately inversely proportional to the average degree, and a higher ratio of average anchor degree to average sensor degree yields better localization accuracy. Furthermore, the paper demonstrates a strong connection between LB-E-AGDOP and the best achievable accuracy. Finally, we validate the theory via numerical simulations with three different random graph models.
\end{abstract}

}

\maketitle

\IEEEdisplaynotcompsoctitleabstractindextext

%
\IEEEpeerreviewmaketitle

%
\ifCLASSOPTIONcompsoc
  \noindent\raisebox{2\baselineskip}[0pt][0pt]%
  {\parbox{\columnwidth}{\section{Introduction}\label{sec:introduction}%
  \global\everypar=\everypar}}%
  \vspace{-1\baselineskip}\vspace{-\parskip}\par
\else
  \section{Introduction}\label{sec:introduction}\par
\fi
%

%
%
%
%
\IEEEPARstart{W}{ireless} sensor networks (WSNs) hold considerable promise for large-scale, flexible, robust, cost-effective data collection and information processing in complex environments
\cite{Akyildiz2002393,Langendoen2003,Heidemann1683469}. Location awareness is a fundamental feature in many WSN applications because ``sensing data without knowing the sensor's location is meaningless'' \cite{Rabaey00}. Location information can also help a node interact with its neighbors and surroundings, improving networking operations such as geographic routing and topology control \cite{Bruck2009}.

To enable location awareness in WSNs, a wide variety of localization schemes have been explored over the past decade. According to the measurements used to estimate locations, these schemes can be generally classified as range-based \cite{MooreRobust04,Teymorian5291229}, angle-based \cite{Peng4068140,Bruck2009}, proximity-based \cite{ShangConnectivity03,Wang4663064}, and event-driven \cite{RomerLighthouse03,ZhongUncontrolledEvents2012}. Besides, the localization schemes can be categorized as either noncooperative or cooperative \cite{Patwari2005,Wymeersch09}. In a noncooperative scheme, the unknown-location nodes (hereinafter referred to as \emph{sensor nodes} or simply \emph{sensors}) make measurements with known-location references (hereinafter referred to as \emph{anchor nodes} or \emph{anchors}), without any communication between sensor nodes. In a cooperative scheme, in addition to anchor-to-sensor measurements, each sensor also makes measurements with neighboring sensors; the additional information gained from sensor-to-sensor measurements enhances localization accuracy, availability, and robustness. In the literature, cooperative localization has also been named as ``relative,'' ``GPS-free,'' ``multi-hop,'' or ``network'' localization.

This paper focuses on range-based cooperative localization schemes.
\begin{figure}[!tb]
  \centering\small
  \includegraphics[width=0.95\linewidth]{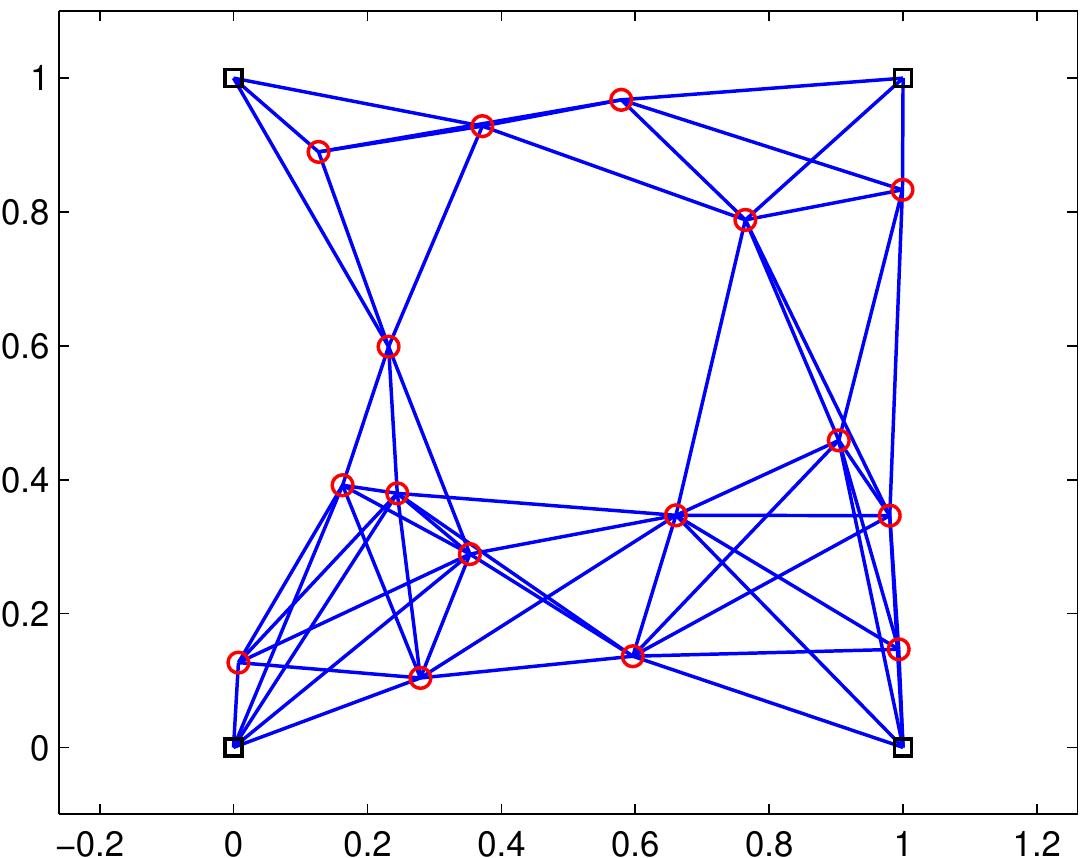}
  \caption{A scenario of range-based cooperative localization in 2 dimensions. Black squares denote anchor nodes whose locations are known, red circles denote sensor nodes whose locations are to be estimated, and blue lines represent ranging links which provide inter-node distance information. Sensors are randomly distributed and ranging links are randomly established according to a certain random graph model. This paper aims for a generalized theory that characterizes the connection between system parameters (namely network connectivity and network size) and localization accuracy.}\label{fig:snrnwk2d}
\end{figure}
As illustrated in Fig.~\ref{fig:snrnwk2d}, anchor nodes (denoted by black squares) are aware of their locations, and sensor nodes (denoted by red circles) determine their locations using inter-node distance information. Range-based cooperative localization is essentially a graph embedding problem \cite{Jackson20051,Aspnes06,Ding2010}. Connectivity of the graph exerts significant influence on many performance metrics, such as accuracy, energy efficiency, localizability, robustness, and scalability. Although localizability has been extensively studied with respect to connectivity \cite{Jackson20051,Aspnes06}, the relationship between accuracy and connectivity has not yet been theoretically treated. The objective of this paper is a generalized theory that quantitatively characterizes the connection between network parameters (namely, network connectivity and network size) and localization accuracy. The theory provides compendious guidelines on the design and deployment of location-aware WSNs.

\subsection{Related work}

%


As previously mentioned, range-based cooperative localization is a graph embedding problem. Saxe \cite{Saxe79} has shown that testing the embeddability of weighted graphs (equivalently, testing localizability) is strongly NP-hard. Aspnes et al.\cite{Aspnes04} have further proven that localization in sparse networks is NP-hard. However, when the network is densely connected such that $O(N^2)$ pairs of nodes know their relative distances, where $N$ is the number of nodes, there are efficient algorithms such as multidimensional scaling (MDS) \cite{Costa2006} and semidefinite programming (SDP) \cite{So2007} for solving the localization problem.

Cooperative localization can also be seen as a high-dimension optimization problem that finds a vector of node locations such that inter-node distances are as close to range measurements as possible. In general, this optimization problem may have many local optimums. The MDS and SDP algorithms \cite{Costa2006,So2007}, as well as some stochastic optimization algorithms \cite{kannan2005simulated}, are able to find a solution close to the global optimum under certain conditions (e.g., dense connectivity). The solution is not necessarily very accurate but can be treated as an initial guess. Then, the location solution can be improved using an iterative algorithm such as \emph{lateration} (also referred to as \emph{trilateration}, \emph{multilateration}, and even mistakenly \emph{triangulation}) \cite{Savvides2001DFL,MooreRobust04,Teymorian5291229,Misra06,Osa13}.

It should be noted that even when the localization problem is overdetermined, noisy range measurements can lead to flip ambiguity in a bad geometry, as discussed in \cite{Dulman08,Moravek12}. The flip ambiguity can cause incorrect initial guess or unconvergence of some iterative lateration algorithms. The localization accuracy discussed in this paper (as well as many previous papers) is about the lateration errors under the assumption that the initial guess is correct and the lateration converges.

%

The accuracy of lateration has been widely studied using the Cram\'{e}r-Rao (CR) bound \cite{Chang1381943,Larsson1268022,Yu2007,Shang1392106,Alsindi2008,Zhang20101391,Penna5581213,Wang2013}, which is the reciprocal of the Fisher information matrix \cite{Rao45}. Many of these studies ended up with a complicated Fisher information matrix (or an equivalent form) involving node locations, and did not give an explicit closed-form expression that can characterize localization accuracy with respect to network connectivity.

There have been some papers (e.g., \cite{Shang1392106,Savvides05}) using Monte-Carlo simulations to reveal that (1) for one sensor node and $m$ randomly-distributed anchor nodes, the CR bound is inversely proportional to $m$; (2) higher percentage of anchors results in better accuracy. However, there is still a dearth of theory to describe these relationships precisely for a more generalized setting.

Two papers \cite{Shen10,Javanmard13} deserve special mentions here because they contain a similar flavor to this paper. Shen et al\@. \cite{Shen10} presented scaling laws of localization accuracy for randomly-deployed nodes. The scaling laws indicate that sensors and anchors ``contribute equally'' to localization accuracy. However, this statement is correct only for dense network (a larger number of nodes); in this paper, we shall show that anchors generally contribute more than sensors do, and the sensors and anchors tend to contribute equally as the number of sensors increases. In \cite{Javanmard13}, Javanmard and Montanari offered neat upper and lower bounds of localization accuracy for random geometric graphs. However, the bounds are only applicable to random geometric graphs, and require range errors to be uniformly bounded.

\subsection{Our contributions}

In this paper, we use the \emph{average degree} as a connectivity metric and use geometric \emph{dilution of precision} (DOP), equivalent to the CR bound, to quantify localization accuracy. We have proved a lower bound on the expectation of average geometric DOP (E-AGDOP) under the assumption that nodes are randomly distributed, and nodes are randomly connected such that the graph of network can reach a certain average degree. We have further derived a closed-form formula that relates the lower bound to only three parameters: average anchor degree, average sensor degree, and number of sensor nodes. The formula shows that (1) localization accuracy is approximately inversely proportional to the average degree, and (2) average anchor degree contributes more to localization accuracy than average sensor degree does. Our numerical examples and simulations have validated the formula and further shown that (1) the lower bound is strongly connected to the best achievable localization accuracy, and (2) the lower bound is applicable to many random graph models.

\subsection{Outline of the paper}

The remainder of this paper is organized as follows. Section~\ref{sec:ProblemFormulation} formulates the cooperative localization problem and introduces the assumptions, definitions, and notations used throughout this paper. Section~\ref{sec:locAccuracy} analyzes localization accuracy and connects it to DOP. Section~\ref{sec:LowerBound} derives a closed-form expression LB-E-AGDOP, which is a function of network connectivity and size. Section~\ref{sec:interpretation} shows the strong connection between LB-E-AGDOP and the best achievable accuracy. Numerical simulation results are presented in Section~\ref{sec:simResults} to validate the theory. Finally, Section~\ref{sec:conclusion} concludes the paper. Proofs of key theorems and equations are provided in Appendices~\ref{ap:1} and \ref{ap:2}.

\section{Preliminaries}\label{sec:ProblemFormulation}

\subsection{Problem formulation}

In this paper, a sensor network is modeled as a \emph{simple} graph\footnote{A simple graph, also known as a strict graph, is an unweighted, undirected graph containing no self-loops or multiple edges \cite{West2001introduction}.} $\mathcal{G}=(V,E)$, where $V=\{1, 2, \ldots, N\}$ is a set of $N$ nodes (or\ vertices), and $E=\{e_1,e_2,\ldots,e_K\}\subseteq V\times V$ is a set of $K$ links (or\  edges) that connect the nodes \cite{Aspnes06}.

All nodes are in a $d$-dimensional Euclidean space ($d\geq 1$), with the locations denoted by $p_n\in \mathbb{R}^d$, $n=1$, \ldots, $N$. The first $N_S$ nodes, labeled 1 through $N_S$, are \emph{sensor} nodes (or\  mobile nodes), whose locations are unknown; the rest $N_A=N-N_S$ nodes, labeled $N_S+1$ through $N$, are \emph{anchor} nodes (or\  beacon nodes). Anchors are aware of their exact locations through built-in GPS receivers or manual pre-programming during deployment.

An unordered pair $e_k=(i_k,j_k)\in E$ if and only if there exists a direct ranging link between nodes $i_k$ and $j_k$. The link provides inter-node distance information $\rho_k=r_k+\epsilon_k$, where $r_k=\|p_{i_k}-p_{j_k}\|$ is the actual Euclidean distance between nodes $i$ and $j$, and $\epsilon_k$ is the range measurement error.


The cooperative localization problem is to determine the locations of sensor nodes $p_n$, $n=1$, \ldots, $N_S$, given a fixed network graph $G$, known locations of anchors $p_n$, $n=N_S+1$, \ldots, $N$, and range measurements $\rho_k$, $k=1$, \ldots, $K$.

\subsection{Assumptions}\label{ssec:assumptions}


\subsubsection{Range measurement errors}

The range measurements $\rho_k$ can be obtained by a variety of methods, such as one-way time of arrival (ToA), two-way ToA, or received signal strength indication (RSSI) \cite{Farooq2012WSN}. One-way ToA usually result in biased range measurements due to unsynchronized clocks \cite{Misra06}, while two-way ToA and RSSI do not depend on clocks. In this paper, we assume zero clock biases in range measurements. Our assumption holds for the cases of two-way ToA, RSSI, and one-way ToA with perfect clock synchronization.

Range measurement errors in RSSI-based methods are usually treated as having a log-normal distribution \cite{Patwari2003}. For most ToA-based methods, line-of-sight range measurement errors can be modeled as zero-mean Gaussian random variables \cite{MooreRobust04,Misra06}. In this paper, we adopt the Gaussian assumption.

\subsubsection{Coordinate symmetry}\label{sssec:cs}

For any link $e_k=(i_k,j_k)\in E$, we assume that the direction vector, defined as
\begin{equation}
v_k=r_k^{-1}(p_{i_k}-p_{j_k})=[v_{k,1},\ldots,v_{k,d}]\T,
\end{equation}
satisfies the following condition:
\begin{equation}
  \E(v_{k,1}^2)=\E(v_{k,2}^2)=\cdots=\E(v_{k,d}^2).
\end{equation}
This assumption holds for all sensor-to-sensor links if all sensor nodes are uniformly distributed in a space that is symmetrical in all coordinates. This assumption holds for anchor-to-sensor links if anchors are uniformly distributed or anchors are fixed at certain special locations such as the scenario shown in Fig.~\ref{fig:snrnwk2d}.

The list below shows three well-studied models of random graphs.
\begin{itemize}
  \item Erd\H{o}s--R\'{e}nyi random graph (ERG) $\mathcal{G}(N,p)$: Nodes are connected randomly regardless of the distance. Each link is included in the graph with probability $p$ independent from every other link \cite{ER59}.
  \item Random geometric graph (RGG) $\mathcal{G}(N,r)$: Two nodes are connected if and only if the distance between them is at most a threshold $r$ \cite{Penrose03,DiazRGG}.
  \item Random proximity graph (RPG) $\mathcal{G}(N,k)$: Each node connects to its $k$ nearest neighbors. These graphs are also denoted $k$-NNG \cite{DiazRGG}.
\end{itemize}
With properly chosen locations of anchor nodes, all of them satisfies the coordinate symmetry condition. Therefore, the theory developed in this paper is applicable to, but not limited to, the above models.



\subsection{Metrics of connectivity}

For all nodes $n=1$, \ldots, $N$, we define the following \emph{degrees}:
\begin{itemize}
  \item Anchor degree: $\dgr_A(n)$, the number of anchor nodes incident to node $n$;
  \item Sensor degree: $\dgr_S(n)$, the number of sensor nodes incident to node $n$;
  \item Degree: $\dgr(n)=\dgr_A(n)+\dgr_S(n)$, the number of nodes incident to node $n$;
\end{itemize}
We assume that there are no anchor-to-anchor links, i.e., $\dgr_A(n)=0$ for $n=N_S+1$, \ldots, $N$, because anchor-to-anchor links are helpless when locations of anchors are perfectly known.

In graph theory, connectivity is usually described by vertex connectivity or edge connectivity: a graph is $\kappa$-vertex/edge-connected if it remains connected whenever fewer than $\kappa$ vertices/edges are removed \cite{DiestelGraphTheory}. Unfortunately, vertex/edge connectivity mainly reflects some ``minimum'' properties of connectivity, such as $\min_{n\in\{1,\ldots,N\}}\dgr(n)$ \cite{DiestelGraphTheory}, and does not distinguish between sensor and anchor nodes. This paper uses average degrees to characterize the overall connectivity of the network. Average degrees are defined as
\begin{equation}
  \delta_*=\frac1{N_S}\sum_{n=1}^{N_S}\dgr_*(n),
\end{equation}
where the subscript ${}_*$ can be blank, ${}_A$, or ${}_S$, for the average degree, average anchor degree, or average sensor degree, respectively. 

Let $K_S$ and $K_A$ denote the number of sensor-to-sensor and anchor-to-sensor links in the network, respectively. It is easy to verify the equalities $K=K_S+K_A$, $N_S\delta_S=2K_S$, $N_S\delta_A=K_A$, and $\delta=\delta_S+\delta_A$.


\subsection{List of notations}

\begin{flushleft}\setlength{\IEEEdlabelindent}{0pt}
\begin{description}[\setlabelwidth{$\mathcal{N}(\mu,\sigma^2)$}]
  \item[$d$] dimensionality
  \item[$\dgr(n)$] degree of node $n$
  \item[$\dgr_A(n)$] anchor degree of node $n$
  \item[$\dgr_S(n)$] sensor degree of node $n$
  \item[$\delta$] average degree
  \item[$\delta_A$] average anchor degree
  \item[$\delta_S$] average sensor degree
  \item[$E$] set of all links, $\{e_1, \ldots, e_K\}$
  \item[$\epsilon_k$] range error of link $e_k$
  \item[$\bm \varepsilon$] localization errors, $\bm\varepsilon=\bm p^{(\infty)}-\bm p$
  \item[$F$] inverse of DOP matrix, $F=G\T G$
  \item[$\mathcal{G}$] graph $(V,E)$
  \item[$G$] geometry matrix
  \item[$H$] DOP matrix, $H=(G\T G)^{-1}$
  \item[$I$] identity matrix
  \item[$i_k$] head of link $e_k=(i_k,j_k)$
  \item[$j_k$] tail of link $e_k=(i_k,j_k)$
  \item[$K$] number of all links
  \item[$K_A$] number of all anchor-to-sensor links
  \item[$K_S$] number of all sensor-to-sensor links
  \item[$L$] Laplacian matrix of graph $\mathcal{G}$, $d\check \Xi=[L_{ij}]_{i,j\in\{1,2,\ldots,N_S\}}$
  \item[$m$] index of dimensions, $m\in\{1,\ldots,d\}$.
  \item[$N$] number of nodes, $N=N_A+N_S$
  \item[$N_A$] number of anchor nodes
  \item[$N_S$] number of sensor nodes
  \item[$\mathcal{N}(\mu,\sigma^2)$] Gaussian distribution with mean $\mu$ and variance $\sigma^2$
  \item[$p_i$] location of node $i$
  \item[$r_k$] actual distance of link $e_k$
  \item[$\rho_k$] distance measurement of link $e_k$
  \item[$\Sigma$] covariance of range errors
  \item[$\sigma_k$] standard deviation of range errors of link $e_k$
  \item[$V$] set of all nodes, $\{1, \ldots, N\}$
  \item[$v_k$] unit vector denoting the direction of link $e_k$, $v_k=r_k^{-1}(p_{i_k}-p_{j_k})$
  \item[$\Xi$] conditional expectation of $F$ given certain links, $\Xi=\E_{\text{locations}}(F|\text{links})$
  \item[$\check \Xi$] submatrix of $\Xi$, representting one coordinate
\medskip
  \item[$\|z\|$] Euclidean norm of vector $z$
  \item[$A\succeq B$] $A-B$ is positive semidefinite
\end{description}
\end{flushleft}

\section{Localization Accuracy}\label{sec:locAccuracy}

Localization is essentially an optimization problem that finds coordinate vectors $p_n\in \mathbb{R}^d$, $n=1$, \ldots, $N_S$, such that for each ranging link $e_k=(i_k,j_k)\in E$, the distance $r_k=\|p_{i_k}-p_{j_k}\|$ is as close to the range measurement $\rho_k$ as possible. 

Assume that range errors follow a zero-mean Gaussian distribution:
\begin{equation}
  \epsilon_k=\rho_k-r_k\sim\mathcal{N}(0,\sigma_k^2),\quad\forall k=1,\ldots,K.
\end{equation}
The maximum-likelihood estimation of $\{p_n\}_{n=1}^{N_S}$ is equivalent to the weighted least squares (LS) problem
\begin{equation}\label{eq:WLS}
\begin{split}
  &\argmax_{\{p_n\}_{n=1}^{N_S}}\Prob\bigl(\{\rho_k\}_{k=1}^K\bigm|\{p_n\}_{n=1}^{N_S}\bigr)\\
  &=\argmax_{\{p_n\}_{n=1}^{N_S}}\prod_{k=1}^K\frac{1}{2\pi\sigma_k^2}\exp\Bigl(-\frac{(\|p_{i_k}-p_{j_k}\|-\rho_k)^2}{2\sigma_k^2}\Bigr)\\
  &=\argmin_{\{p_n\}_{n=1}^{N_S}}\sum_{k=1}^K\frac{(\|p_{i_k}-p_{j_k}\|-\rho_k)^2}{\sigma_k^2}.
\end{split}
\end{equation}

The LS problem cannot be directly solved because the distance $r_k=\|p_{i_k}-p_{j_k}\|$ is a nonlinear function of the coordinate vectors $p_{i_k}$ and $p_{j_k}$. Let $\bm r=(r_{1},r_{2},\ldots,r_{K})\T\in\mathbb{R}^K$ and $\bm p=\operatorname{column}\{p_{1},p_{2},\ldots,p_{N_S}\}\allowbreak\in\mathbb{R}^{dN_S}$. The first-order linear approximation of the distance function $\bm r(\bm p^{(0)}+\Delta\bm p)$ with respect to an initial guess $\bm p^{(0)}$ can be written as
\begin{equation}
  \bm r(\bm p^{(0)}+\Delta\bm p)=\bm r(\bm p^{(0)})+G\Delta\bm p,
\end{equation}
where the \emph{geometry matrix} $G\in\mathbb{R}^{K\times dN_S}$ is given by
\begin{equation}\label{eq:G}
\begin{split}
G&=\frac{\partial\bm r}{\partial\bm p}\\
&=\begin{bmatrix}
   \frac{\partial r_1}{\partial p_{1,1}} & \ldots & \frac{\partial r_1}{\partial p_{1,d}} & \ldots & \frac{\partial r_1}{\partial p_{N_S,1}} & \ldots & \frac{\partial r_1}{\partial p_{N_S,d}} \\
   \vdots & \vdots & \vdots & \vdots & \vdots & \vdots \\
   \frac{\partial r_K}{\partial p_{1,1}} & \ldots & \frac{\partial r_K}{\partial p_{1,d}} & \ldots & \frac{\partial r_K}{\partial p_{N_S,1}} & \ldots & \frac{\partial r_K}{\partial p_{N_S,d}} \\
 \end{bmatrix},
\end{split}
\end{equation}
where $p_{i,m}$, $m=1$, \ldots, $d$, is the $m$th element of the coordinate vector $p_i$. Each element of the geometry matrix $G$ is given by
\begin{equation}\label{eq:Gij}
\begin{split}
G_{k,(n-1)d+m}
&=\frac{\partial r_k}{\partial p_{n,m}}
=\frac{\partial \|p_{i_k}-p_{j_k}\|}{\partial p_{n,m}}\\
&\hspace{-3em}=
\begin{cases}
  \frac{p_{i_k,m}-p_{j_k,m}}{\|p_{i_k}-p_{j_k}\|}=v_{k,m} & \text{if }n=i_k,\\
  \frac{p_{j_k,m}-p_{i_k,m}}{\|p_{i_k}-p_{j_k}\|}=-v_{k,m} & \text{if }n=j_k,\\
  0 & \text{otherwise}.\\
\end{cases}
\end{split}
\end{equation}
Each row of $G$ represents a link. There are only $d$ nonzero elements in a row for an anchor-to-sensor link, and there are $2d$ nonzero elements for an sensor-to-sensor link. Given that each row of $G$ has $dN_S$ elements, $G$ is highly sparse when the network contains many sensor nodes.

When the network is localizable, $G$ must be a tall matrix (i.e., $K\geq dN_S$ \cite{Jackson20051,Aspnes06,HengIPSN13}) with full column rank. Then, the weighted LS problem \eqref{eq:WLS} can be solved by the following iterative algorithm based on the Newton--Raphson method \cite{Misra06}:
\begin{equation}\label{eq:iterative_algorithm}
  \bm p^{(n+1)}=\bm p^{(n)}+(G\T \Sigma^{-1}G)^{-1}G\T \Sigma^{-1}[\bm\rho-\bm r(\bm p^{(n)})],
\end{equation}
where $\bm \rho=(\rho_{1},\rho_{2},\ldots,\rho_{K})\T$, 
$\Sigma=\Cov(\bm\epsilon,\bm\epsilon)$ is the covariance of range errors, where $\bm\epsilon=(\epsilon_1,\ldots,\epsilon_K)\T$.

When the initial guess $\bm p^{(0)}$ is accurate enough and the iteration converges, the localization errors $\bm\varepsilon$ have the following relationship to the range errors $\bm\epsilon=(\epsilon_1,\ldots,\epsilon_K)\T$:
\begin{equation}
\begin{split}
  \bm\varepsilon&=\bm p^{(\infty)}-\bm p
  =(G\T \Sigma^{-1}G)^{-1}G\T \Sigma^{-1}(\bm\rho-\bm r)\\
  &=(G\T \Sigma^{-1}G)^{-1}G\T \Sigma^{-1}\bm\epsilon.
\end{split}
\end{equation}
The covariance of localization errors is thus given by
\begin{equation}
\begin{split}
  \Cov(\bm\varepsilon,\bm\varepsilon)
  &=(G\T \Sigma^{-1}G)^{-1}G\T \Sigma^{-1}\Cov(\bm\epsilon,\bm\epsilon)\\
  &\qquad\Sigma^{-1}G\T(G\T \Sigma^{-1}G)^{-1}\\
  &=(G\T \Sigma^{-1}G)^{-1}.
\end{split}
\end{equation}
This has achieved the CR bound \cite{Chang1381943,Shang1392106,Alsindi2008,Zhang20101391,Larsson1268022,Penna5581213}.

If range measurement errors are independent and identically distributed (iid), i.e., $\Sigma=\operatorname{diag}(\sigma^2,\ldots,\sigma^2)$, we have
\begin{equation}\label{eq:covLocErr}
  \Cov(\bm\varepsilon,\bm\varepsilon)=(G\T \Sigma^{-1}G)^{-1}=\sigma^2(G\T G)^{-1}.
\end{equation}
The matrix $H=(G\T G)^{-1}\in\mathbb{R}^{dN_S\times dN_S}$ is referred to as dilution of precision (DOP) matrix. DOP is a term widely used in satellite navigation specifying the multiplicative effect on positioning accuracy due to satellite geometry\footnote{The DOP is usually defined in the form of $\sqrt{\tr[(G\T G)^{-1}]}$ \cite{Savvides05,Misra06}. In this paper, we define DOP in the form of $\tr[(G\T G)^{-1}]$ for simplicity in calculation and analysis.} \cite{Misra06}. For cooperative localization, DOP specifies the multiplicative effect due to not only geometry of the nodes but also connectivity of the network. DOP decouples localization accuracy from range accuracy. The smaller DOP is, the better localization accuracy one would expect.

A diagonal element $H_{(n-1)d+m,(n-1)d+m}$ is the DOP of coordinate $m$ for node $n$. The sum of all the diagonal elements, $\tr(H)$, is the geometric DOP (GDOP) of the whole network. In this paper, we define average GDOP (AGDOP) as GDOP divided by the number of sensor nodes, $\tr(H)/N_S$. AGDOP is a performance indicator of localization accuracy due to network geometry and connectivity.

For a network where nodes are deployed and connected randomly, AGDOP is a random variable. The expectation of AGDOP (E-AGDOP) indicates the expected localization accuracy because the root-mean-square localization error is proportional to $\sqrt{\text{E-AGDOP}}$. We shall use E-AGDOP and its lower bound to study the relationship between localization accuracy and network connectivity in the rest of the paper.



\section{Lower Bound on E-AGDOP}\label{sec:LowerBound}

In this section, we shall prove that $[\E(G\T G)]^{-1}$ is a lower bound on $\E[(G\T G)^{-1}]$. Furthermore, we shall show that it is possible to evaluate this lower bound analytically for a random network (randomly-deployed nodes and randomly-established links) that achieves a certain level of connectivity.

%
\begin{theorem}[Lower bound on DOP matrix]\label{thm:H_tildeH}
For a random network with a non-singular geometry matrix $G$ defined in \eqref{eq:G},
\begin{equation}
\E[(G\T G)^{-1}]\succeq [\E (G\T G)]^{-1},
\end{equation}
where the operator $X\succeq Y$ denotes that $X-Y$ is positive semidefinite.
\end{theorem}
\begin{proof}
Detailed in Appendix~\ref{ap:1}.
\end{proof}

The matrix $F=G\T G$ is a function of node locations and links, both of which have been assumed to be random. Let us calculate $\E F$ by the following two steps:
\begin{enumerate}
  \item $\Xi=\E_{\text{nodes}}(F|\text{links})$, conditional expectation of $F$ for randomly-deployed nodes given certain links;
  \item $\E F=\E_{\text{links}}(\Xi)$, expectation of $\Xi$ for randomly-established links.
\end{enumerate}

\subsection{Step 1: randomly-deployed nodes}

Recall \eqref{eq:Gij} which describes the elements in $G$. Note that when link $e_k$ connects to node $n$, i.e., $n\in\{i_k,j_k\}$,
\begin{equation}\label{eq:sum_dr_dp}
  \sum_{m=1}^d\Bigl(\frac{\partial r_k}{\partial p_{n,m}}\Bigr)^2=
  \frac{\sum_{m=1}^d(p_{i_k,m}-p_{j_k,m})^2}{\|p_{i_k}-p_{j_k}\|^2}=1.
\end{equation}
By the coordinate symmetry assumption (Section~\ref{sssec:cs}), we have
\begin{equation}
  \E\Bigl(\frac{\partial r_k}{\partial p_{n,1}}\Bigr)^2=
  \E\Bigl(\frac{\partial r_k}{\partial p_{n,2}}\Bigr)^2=\cdots=
  \E\Bigl(\frac{\partial r_k}{\partial p_{n,m}}\Bigr)^2.
\end{equation}
To satisfy \eqref{eq:sum_dr_dp}, we must have
\begin{equation}\label{eq:Gij2}
  \E\Bigl(\frac{\partial r_k}{\partial p_{n,m}}\Bigr)^2=\frac1d,\quad\forall m=1,\ldots,d.
\end{equation}


Therefore, the elements of matrix $F=\{F_{\tilde i\tilde j}\}\in\mathbb{R}^{dN_S\times dN_S}$ have the conditional expectation
\begin{equation}
\begin{split}
\Xi_{\tilde i\tilde j}&=\E_{\text{nodes}}(F_{\tilde i\tilde j}|\text{links})
=\E\sum_{k=1}^K\frac{\partial r_k}{\partial p_{i,m_1}}\frac{\partial r_k}{{\partial p_{j,m_2}}}\\
&=\begin{cases}
  \frac1d\dgr(i) & \text{if }i=j\text{ and }m_1=m_2,\\
  -\frac1d & \text{if }(i, j)\in E\text{ and }m_1=m_2,\\
  0 & \text{otherwise},\\
\end{cases}
\end{split}
\end{equation}
where $\tilde i=(i-1)d+m_1$, $\tilde j=(j-1)d+m_2$, $1\leq m_1,m_2\leq d$. For instance, let us consider a very simple sensor network shown in Fig.~\ref{fig:simpleSN}.
The matrix $\Xi$ for this network is given by
\begin{equation}\label{eq:Xi}
\Xi_{\text{Fig.~\ref{fig:simpleSN}}}=
\begin{bmatrix}
  {\color{red}\frac12} & 0 & {\color{red}-\frac12} & 0 & {\color{red}0} & 0 \\
  0 & {\color{blue}\frac12} & 0 & {\color{blue}-\frac12} & 0 & {\color{blue}0} \\
  {\color{red}-\frac12} & 0 & {\color{red}1} & 0 & {\color{red}-\frac12} & 0 \\
  0 & {\color{blue}-\frac12} & 0 & {\color{blue}1} & 0 & {\color{blue}-\frac12} \\
  {\color{red}0} & 0 & {\color{red}-\frac12} & 0 & {\color{red}1} & 0 \\
  0 & {\color{blue}0} & 0 & {\color{blue}-\frac12} & 0 & {\color{blue}1} \\
\end{bmatrix}.
\end{equation}
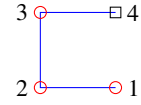
\begin{figure}[!tb]
  \centering\small
\begin{tikzpicture}[every rectangle node/.style={inner sep = 2pt, black}, every circle node/.style={inner sep = 1.6pt, red}, every path/.style={blue}]
\draw (1,0) node[circle,draw,label=right:1] {} -- (0,0) node[circle,draw,label=left:2] {} -- (0,1) node[circle,draw,label=left:3] {} -- (1,1) node[rectangle,draw,label=right:4] {};
\end{tikzpicture}
  \caption{A simple sensor network comprised of 4 nodes and 3 links in 2 dimensions. Nodes 1 to 3 are sensors; node 4 is an anchor ($N_S=3$, $N_A=1$, $K_S=2$, $K_A=1$). Eq.~\eqref{eq:Xi} shows the matrix $\Xi$ for this network.}\label{fig:simpleSN}
\end{figure}

As shown by the red- and blue-colored elements in \eqref{eq:Xi}, we have $\Xi=\check \Xi\otimes I$, where $\otimes$ denotes the Kro\-neck\-er product, and $I$ is the identity matrix of size $d$. The elements of the matrix $\check \Xi\in\mathbb{R}^{N_S\times N_S}$ are given by
\begin{equation}
\check \Xi_{ij}=
\begin{cases}
  \frac1d\dgr(i) & \text{if }i=j,\\
  -\frac1d & \text{if }(i,j)\in E,\\
  0 & \text{otherwise}.\\
\end{cases}
\end{equation}
For the sensor network shown by Fig.~\ref{fig:simpleSN}, the matrix $\check \Xi$ is given by
\begin{equation}
\check\Xi_{\text{Fig.~\ref{fig:simpleSN}}}=
\begin{bmatrix}
  \frac12 & -\frac12 & 0  \\
  -\frac12 & 1 & -\frac12 \\
  0  & -\frac12  & 1  \\
\end{bmatrix}.
\end{equation}

The matrix $\check\Xi$ (as well as $\Xi$) indicates a relationship between localization accuracy and graph Laplacians \cite{Anderson85}. Let $L$ denote the Laplacian matrix of the graph $\mathcal{G}$. It can be seen that $\check \Xi=d^{-1}[L_{ij}]_{i,j\in\{1,2,\ldots,N_S\}}$, i.e., $d\check \Xi$ is the matrix of $L$ obtained by deleting its last $N_A$ rows and columns that are related to the anchor nodes.

The lower bound on E-AGDOP (LB-E-AGDOP) can be calculated by inverting $F=\E \Xi$ or, equivalently, inverting $\check F=\E \check \Xi$, because $\tr[(\E \Xi)^{-1}]=d\tr[(\E\check \Xi)^{-1}]$.

\subsection{Step 2: randomly-established links}

Given an average degree $\delta$, the trace of $\check \Xi$ is given by
\begin{equation}
  \tr(\check \Xi)=\sum_{i=1}^{N_S}\check \Xi_{ii}=\sum_{i=1}^{N_S}\dgr(i)/d=N_S\delta/d.
\end{equation}
Given an average sensor degree $\delta_S$, there are $K_S=N_S\delta_S/2$ sensor-to-sensor links in the network, and thus $\check \Xi$ includes $N_S\delta_S$ off-diagonal elements with a non-zero value of $-1/d$. Assume that the sensor-to-sensor links are chosen uniformly at random from the set $\{(i,j)|1\leq i<j\leq N_S,i,j\in\mathbb{Z}\}$. Then, each off-diagonal element $\check \Xi_{ij}$, $i\neq j$ satisfies the Bernoulli distribution
\begin{equation}
  \check \Xi_{ij}=
  \begin{cases}
    -1/d & \text{with probability }\frac{\delta_S}{N_S-1},\\
    0 & \text{with probability }1-\frac{\delta_S}{N_S-1}.
  \end{cases}
\end{equation}
Then, the expectation of $\check F$ is given by
\begin{equation}\label{eq:EFij}
\E \check F_{ij}=\E_{\text{links}}(\Xi_{ij})=
\begin{cases}
  \frac{\delta}d & \text{if }i=j,\\
  -\frac{\delta_S}{d(N_S-1)} & \text{otherwise}.\\
\end{cases}
\end{equation}

Appendix~\ref{ap:2} shows that
\begin{equation}\label{eq:checkF}
\tr[(\E\check F)^{-1}]=\frac{N_S}{\eta}\Bigl(1+\frac{\zeta}{1-N_S\zeta}\Bigr),
\end{equation}
where $\eta=d^{-1}[\delta+\delta_S/(N_S-1)]$ and $\zeta=\delta_S/[\delta(N_S-1)+\delta_S]$. Therefore, LB-E-AGDOP is given by
\begin{equation}\label{eq:averageGDOP2}
\begin{split}
\text{LB-E-AGDOP}&=\frac{\tr[(\E F)^{-1}]}{N_S}
=\frac{d\tr[(\E\check F)^{-1}]}{N_S}\\
&=\frac{d}{\eta}\Bigl(1+\frac1{\zeta^{-1}-N_S}\Bigr)\\
&=\frac{d^2}{\delta+\delta_S/(N_S-1)}\Bigl(1+\frac{\delta_S}{\delta_A(N_S-1)}\Bigr)\\
&=\frac{d^2}{\delta}\frac{N_S-1+\delta_S/\delta_A}{N_S-1+\delta_S/\delta}.
\end{split}
\end{equation}



Thus far, we have obtained a closed-form expression for LB-E-AGDOP. It depends on two parameters of network connectivity, $\delta_S$ and $\delta_A$ (note $\delta=\delta_S+\delta_A$), and one parameter of network size, $N_S$. In \eqref{eq:averageGDOP2}, the first term ${d^2}/{\delta}$ shows that LB-E-AGDOP is approximately inversely proportional to the average degree, and grows quadratically with dimensionality. The second term $(N_S-1+\delta_S/\delta_A)\big/(N_S-1+\delta_S/\delta)$ shows that a higher ratio of average anchor degree to average sensor degree leads to better localization accuracy. However, this effect diminishes when number of sensor nodes increase. As $N_S\rightarrow\infty$, the LB-E-AGDOP approaches ${d^2}/{\delta}$ regardless of the ratio of average anchor degree to average sensor degree.


\section{LB-E-AGDOP and the Best Achievable Accuracy}\label{sec:interpretation}

The previous section has proven that the LB-E-AGDOP given by \eqref{eq:averageGDOP2} is a lower bound on E-AGDOP. In this section, we shall show that LB-E-AGDOP describes the best achievable accuracy with certain network connectivity. We first prove that LB-E-AGDOP is equal to the minimum AGDOP for one sensor node. For multiple sensor nodes, we use several numerical examples to show that LB-E-AGDOP is less than and very close to the minimum AGDOP.

\subsection{Minimum AGDOP for one sensor node}

Let us first consider the simplest case that there is only one sensor node, i.e., $N_S=1$. By \eqref{eq:averageGDOP2}, the lower bound becomes
\begin{equation}\label{eq:LB_E_AGDOP1}
\text{LB-E-AGDOP}=\frac{d^2}{\delta}\frac{\delta_S/\delta_A}{\delta_S/\delta}=\frac{d^2}{\delta_A}=\frac{d^2}{N_A}.
\end{equation}
In this subsection, we shall show that the lower bound ${d^2}/{N_A}$ represents the best achievable performance.

When $N_S=1$, the geometry matrix can be written as
\begin{equation}
G=
\begin{bmatrix}
  \cos\theta_1 & \sin\theta_1 \\
  \cos\theta_2 & \sin\theta_2 \\
  \vdots & \vdots \\
  \cos\theta_{N_A} & \sin\theta_{N_A} \\
\end{bmatrix},
\end{equation}
where $\theta_i$, $i=1$, \ldots, $N_A$ is the angular coordinate of anchor node $i$ in the polar coordinate system poled at the sensor node. Fig.~\ref{fig:1S3A} shows a scenario of $N_A=3$.

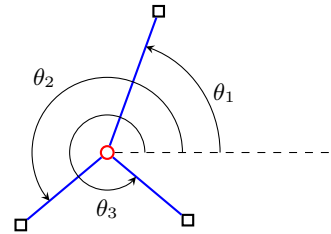
\begin{figure}[!b]
  \centering\small
\begin{tikzpicture}[scale=1,>=stealth,every rectangle node/.style={thick,black,inner sep=2},every circle node/.style={draw,thick,red,inner sep=1.7}]\small
\node [draw,rectangle] at (220:1.5) (A1) {};
\node [draw,rectangle] at (320:1.4) (A2) {};
\node [draw,rectangle] at (70:2) (A3) {};
\node [circle] at (0,0) (S1) {};
\draw [thick,blue] (A1) -- (S1);
\draw [thick,blue] (A2) -- (S1);
\draw [thick,blue] (A3) -- (S1);
\draw [dashed] (S1) -- (3,0);
\draw [->] (1.5,0) arc[start angle=0,delta angle=70,radius=1.5];
\draw [->] (1,0) arc[start angle=0,delta angle=220,radius=1];
\draw [->] (0.5,0) arc[start angle=0,delta angle=320,radius=0.5];
\node at (27:1.75) {$\theta_1$};
\node at (130:1.3) {$\theta_2$};
\node at (270:0.75) {$\theta_3$};
\end{tikzpicture}
\caption{Geometry of one sensor node and three anchor nodes.}\label{fig:1S3A}
\end{figure}

Noting that
\begin{equation}
G\T G=
\begin{bmatrix}
  \sum_{i=1}^{N_A}\cos^2\theta_i & \sum_{i=1}^{N_A}\cos\theta_i\sin\theta_i \\
  \sum_{i=1}^{N_A}\cos\theta_i\sin\theta_i & \sum_{i=1}^{N_A}\sin^2\theta_i \\
\end{bmatrix},
\end{equation}
we obtain a closed-form expression of $(G\T G)^{-1}$ as
\begin{equation}
\begin{split}
&(G\T G)^{-1}=\frac1{\det(G\T G)}\cdot\\
&\quad\begin{bmatrix}
  \sum_{i=1}^{N_A}\sin^2\theta_i & -\sum_{i=1}^{N_A}\cos\theta_i\sin\theta_i \\
  -\sum_{i=1}^{N_A}\cos\theta_i\sin\theta_i & \sum_{i=1}^{N_A}\cos^2\theta_i \\
\end{bmatrix},
\end{split}
\end{equation}
where the determinant of $G\T G$ is given by
\begin{equation}
\begin{split}
\det(G\T G)&=\biggl(\sum_{i=1}^{N_A}\cos^2\theta_i\biggr)\biggl(\sum_{i=1}^{N_A}\sin^2\theta_i\biggr)\\
&\qquad{}-\biggl(\sum_{i=1}^{N_A}\cos\theta_i\sin\theta_i\biggr)^2\\
&=\sum_{i=1}^{N_A}\sum_{j=1}^{N_A}(\cos^2\theta_i\sin^2\theta_j\\
&\qquad{}-\cos\theta_i\sin\theta_i\cos\theta_j\sin\theta_j)\\
&=\sum_{i=1}^{N_A}\sum_{j=1}^{N_A}\cos\theta_i\sin\theta_j\sin(\theta_j-\theta_i)\\
&=\sum_{1\leq i<j\leq N_A}\sin^2(\theta_j-\theta_i).
\end{split}
\end{equation}
Therefore, the GDOP is given by
\begin{equation}\label{eq:trGG0}
\tr[(G\T G)^{-1}]=\frac{N_A}{\sum_{1\leq i<j\leq N_A}\sin^2(\theta_j-\theta_i)},
\end{equation}
which agrees with the result obtained in \cite{Spirito01}.

From \eqref{eq:trGG0} we can find the minimum AGDOP for one sensor node and $N_A$ anchor nodes. Since the sum $S=\sum_{1\leq i<j\leq N_A}\sin^2(\theta_j-\theta_i)$ is continuous differentiable and bounded, a maximum or a minimum is attained if and only if
\begin{equation}
\begin{split}
\frac{\partial S}{\partial\theta_i}
&=\sum_{j=1}^{N_A}\sin[2(\theta_i-\theta_j)]\\
&=0,\qquad \forall i\in\{1,\ldots,N_A\}.
\end{split}
\end{equation}

It is easy to verify that the values $\theta_i=2\pi i/N_A$ satisfy the above condition. The corresponding maximum of $S$ is given by
\begin{equation}
\begin{split}
\sum_{1\leq i<j\leq N_A}\sin^2(\theta_j-\theta_i)
&=\sum_{1\leq i<j\leq N_A}\sin^2\frac{(j-i)2\pi}{N_A}\\
&=\frac12\sum_{i=1}^{N_A}\sum_{j=1}^{N_A}\sin^2\frac{(j-i)2\pi}{N_A}\\
&=\frac12\sum_{i=1}^{N_A}N_A/2=\frac{N_A^2}{4}.\\
\end{split}
\end{equation}
Therefore, the minimum AGDOP is $4/N_A$, equal to the lower bound given by \eqref{eq:LB_E_AGDOP1}.

Fig.~\ref{fig:DOP1} compares the LB-E-AGDOP calculated from \eqref{eq:averageGDOP2} and the AGDOP obtained from simulations with the parameters $N_S=1$, $N_A=3, 4, \ldots, 9$. The green solid curve shows the LB-E-AGDOP. The box-and-whisker plot \cite{Tukey77} shows the sample minimum, lower quartile, median, upper quartile, and sample maximum of the AGDOP. The red plus marks denote statistical outliers. It can be seen that our LB-E-AGDOP is equal to the sample minimum, and is closer to the median when $N_A$ is greater.

\begin{figure}[!h]
  \centering\small
  \includegraphics[width=1\linewidth]{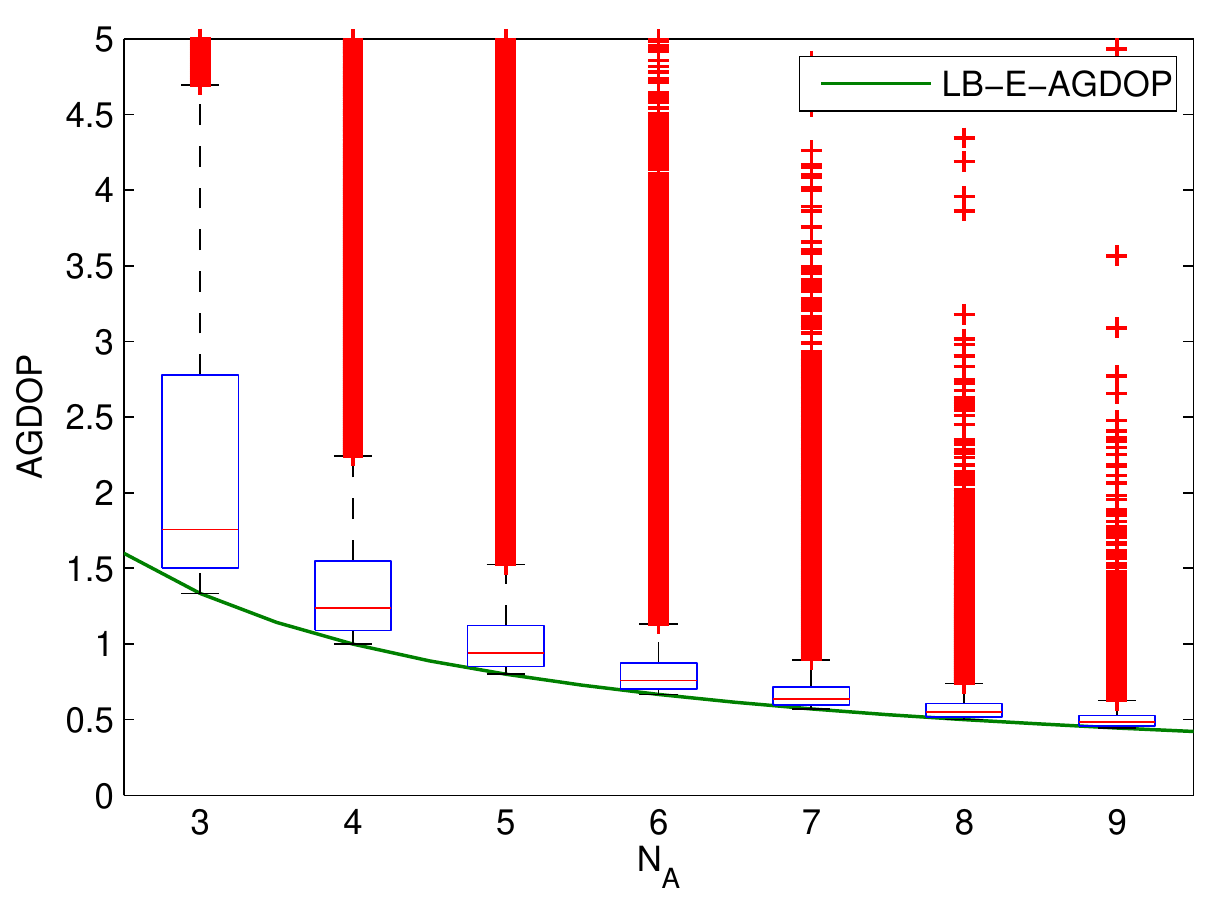}
  \caption{Comparison between LB-E-AGDOP from \eqref{eq:averageGDOP2} and AGDOP from simulations with the parameters $N_S=1$, $N_A=3, 4, \ldots, 9$. The green solid curve shows the LB-E-AGDOP. The box-and-whisker plot shows the sample minimum (lower whisker, black), lower quartile (lower edge of the box, blue), median (central mark, red), upper quartile (upper edge of the box, blue), and sample maximum (upper whisker, black; outliers excluded). Our LB-E-AGDOP is equal to the sample minimum, and is closer to the median when $N_A$ is greater.}\label{fig:DOP1}
\end{figure}

\subsection{Minimum AGDOP for multiple sensor nodes}

For multiple sensor nodes, the minimum AGDOP can be hardly derived from a theoretical analysis. Instead, we present several results obtained from numerical optimization. From the following results it can be seen that for multiple sensor nodes, our LB-E-AGDOP is less than the minimum achievable AGDOP. The gap between LB-E-AGDOP and the minimum AGDOP is smaller when the LB-E-AGDOP is smaller.

\subsubsection*{Case 1: $N_S=2$, $\delta_S=1$, $\delta_A=2$}
One of the best geometries is given by
\begin{center}
\begin{tikzpicture}[every rectangle node/.style={inner sep = 2pt, black}, every circle node/.style={inner sep = 1.6pt, red}, every path/.style={blue}]
\node[circle,draw] (S1) at (0,0) {};
\node[circle,draw] (S2) at (1,0) {};
\node[rectangle,draw] (A1) at (130.662:1) {};
\node[rectangle,draw] (A2) at (-130.662:1) {};
\node[rectangle,draw] (A3) at ($(1,0)+(49.338:1)$) {};
\node[rectangle,draw] (A4) at ($(1,0)+(-49.338:1)$) {};
\draw (A1) -- (S1) -- (A2);
\draw (A3) -- (S2) -- (A4);
\draw (S1) -- (S2);
\draw[<->,green!50!black] (130.662:0.5) arc[start angle=130.662, end angle=229.338, x radius=0.5, y radius=0.5];
\node[left] at (-0.5,0) {\small\color{green!50!black}98.68$^\circ$};
\end{tikzpicture}
\end{center}
The geometry has up-down and left-right reflection symmetry. The minimum AGDOP is 1.633, greater than the LB-E-AGDOP 1.500.

\subsubsection*{Case 2: $N_S=2$, $\delta_S=1$, $\delta_A=3$}
One of the best geometries is given by
\begin{center}
\begin{tikzpicture}[every rectangle node/.style={inner sep = 2pt, black}, every circle node/.style={inner sep = 1.6pt, red}, every path/.style={blue}]
\node[circle,draw] (S1) at (0,0) {};
\node[circle,draw] (S2) at (2,0) {};
\node[rectangle,draw] (A1) at (67.41:1) {};
\node[rectangle,draw] (A2) at (180:1) {};
\node[rectangle,draw] (A3) at (-67.41:1) {};

\node[rectangle,draw] (A4) at ($(2,0)+(112.59:1)$) {};
\node[rectangle,draw] (A5) at ($(2,0)+(0:1)$) {};
\node[rectangle,draw] (A6) at ($(2,0)+(-112.59:1)$) {};

\draw (A1) -- (S1) -- (A2);
\draw (A4) -- (S2) -- (A5);
\draw (A3) -- (S1) -- (S2) -- (A6);

\draw[<->,green!50!black] (67.41:0.5) arc[start angle=67.41, end angle=180, x radius=0.5, y radius=0.5];
\node[left] at (100:0.5) {\small\color{green!50!black}112.59$^\circ$};
\end{tikzpicture}
\end{center}
The geometry has up-down and left-right reflection symmetry. The minimum AGDOP is 1.124, slightly greater than the LB-E-AGDOP 1.067.

\subsubsection*{Case 3: $N_S=3$, $\delta_S=2$, $\delta_A=1$}
One of the best geometries is given by
\begin{center}
\begin{tikzpicture}[every rectangle node/.style={inner sep = 2pt, black}, every circle node/.style={inner sep = 1.6pt, red}, every path/.style={blue}]
\node[circle,draw] (S1) at (0,0) {};
\node[circle,draw] (S2) at (-60:1) {};
\node[circle,draw] (S3) at (-120:1) {};

\node[rectangle,draw] (A1) at (0:1) {};
\node[rectangle,draw] (A2) at ($(-60:1)+(-120:1)$) {};
\node[rectangle,draw] (A3) at (-1,0) {};

\draw (A1) -- (S1) -- (S2);
\draw (A2) -- (S2) -- (S3);
\draw (A3) -- (S3) -- (S1);

\draw[<->,green!50!black] (-60:0.5) arc[start angle=-60, end angle=0, x radius=0.5, y radius=0.5];
\node[right] at (-45:0.5) {\small\color{green!50!black}60.00$^\circ$};
\end{tikzpicture}
\end{center}
The geometry has $\pm 120^\circ$ rotational symmetry. The minimum AGDOP is 2.667, greater than the LB-E-AGDOP 2.000. It should be noted that the connectivity of this case does not ensure unique localizability \cite{Aspnes06,Teymorian5291229}.

\subsubsection*{Case 4: $N_S=3$, $\delta_S=2$, $\delta_A=2$}
One of the best geometries is given by
\begin{center}
\begin{tikzpicture}[every rectangle node/.style={inner sep = 2pt, black}, every circle node/.style={inner sep = 1.6pt, red}, every path/.style={blue}]
\node[circle,draw] (S1) at (0,0) {};
\node[circle,draw] (S2) at (-120:1) {};
\node[circle,draw] (S3) at (-60:1) {};

\node[rectangle,draw] (A1) at (37.925:1) {};
\node[rectangle,draw] (B1) at (142.075:1) {};
\node[rectangle,draw] (A2) at ($(157.925:1)+(-120:1)$) {};
\node[rectangle,draw] (B2) at ($(262.075:1)+(-120:1)$) {};
\node[rectangle,draw] (A3) at ($(-82.075:1)+(-60:1)$) {};
\node[rectangle,draw] (B3) at ($(22.075:1)+(-60:1)$) {};

\draw (A1) -- (S1) -- (S3) -- (B3);
\draw (A2) -- (S2) -- (S1) -- (B1);
\draw (A3) -- (S3) -- (S2) -- (B2);

\draw[<->,green!50!black] (37.925:0.5) arc[start angle=37.925, end angle=142.075, x radius=0.5, y radius=0.5];
\node[above] at (90:0.5) {\small\color{green!50!black}104.15$^\circ$};
\end{tikzpicture}
\end{center}
The geometry has left-right reflection symmetry and $\pm 120^\circ$ rotational symmetry. The minimum AGDOP is 1.313, slightly greater than the LB-E-AGDOP 1.200.
\begin{table}[!t]
  \renewcommand{\arraystretch}{1.5}
  \caption{LB-E-AGDOP establishes a lower bound on AGDOP}
  \label{tab:minimumAGDOP}
  \centering
\begin{tabular}{cccc|cc}
  \hline
  $N_S$ & $N_A$ & $\delta_S$ & $\delta_A$ & minimum AGDOP & LB-E-AGDOP \\ \hline
  1 & $n$ & $0$ & $n$ & $4/n$ & $4/n$ \\
  2 & 4 & 1 & 2 & 1.633 & 1.500 \\
  2 & 6 & 1 & 3 & 1.124 & 1.067 \\
  3 & 3 & 2 & 1 & 2.667 & 2.000 \\
  3 & 6 & 2 & 2 & 1.313 & 1.200 \\
  \hline
\end{tabular}
\end{table}

Table~\ref{tab:minimumAGDOP} summaries the minimum AGDOP and LB-E-AGDOP calculated for the cases mentioned in this section. Although the LB-E-AGDOP is derived from the expectation of AGDOP, we observe that LB-E-AGDOP may also establish a lower bound on AGDOP. 

\begin{figure*}[!t]
  \centering\small
  \includegraphics[width=1\linewidth]{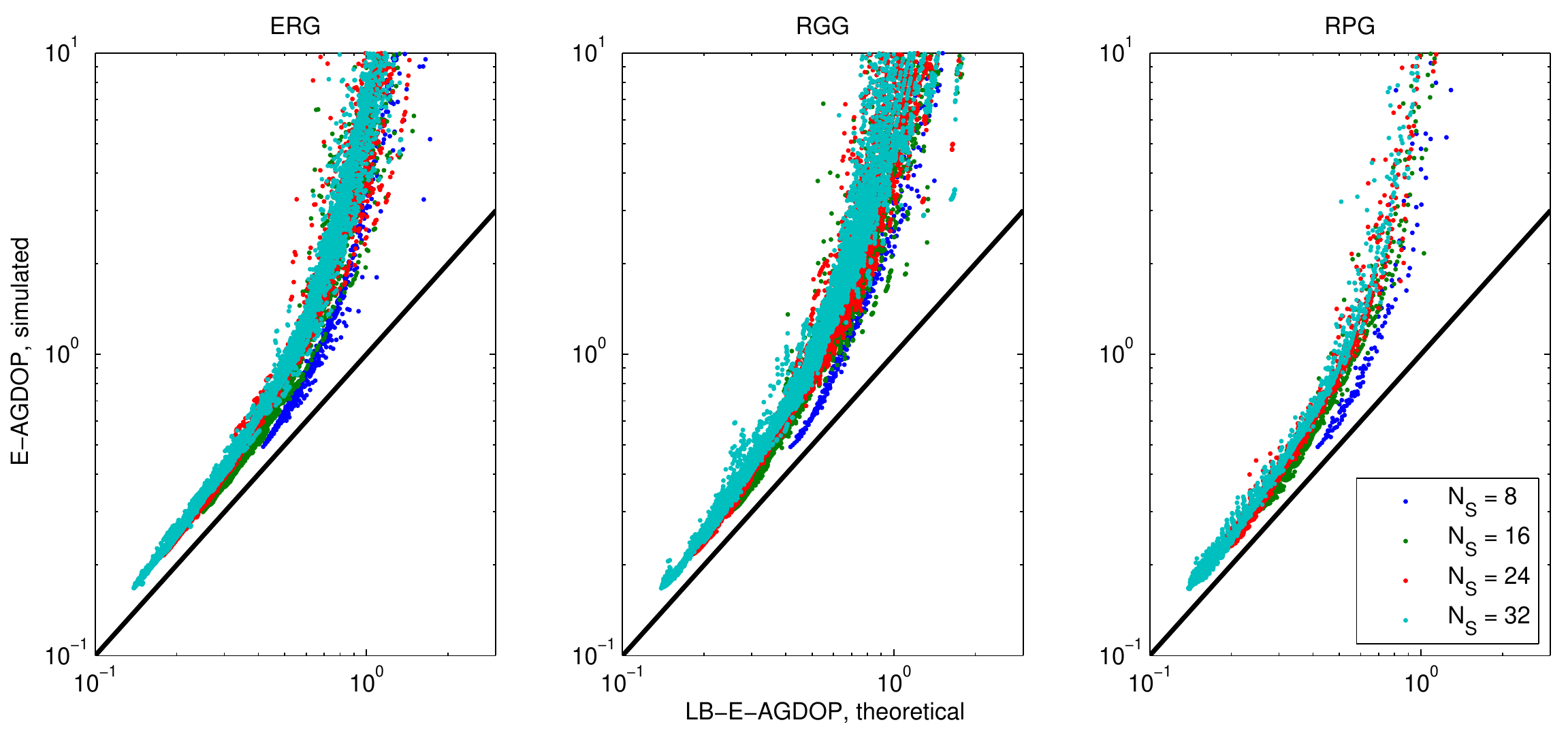}
  \caption{Comparison between the LB-E-AGDOP from \eqref{eq:averageGDOP2} and the \textbf{sample mean} of AGDOP from simulations with the parameters $N_S=8, 16, 24, 32$. The black line shows $y=x$.}\label{fig:DOP_DOP3}
\end{figure*}

\begin{figure*}[!t]
  \centering\small
  \includegraphics[width=1\linewidth]{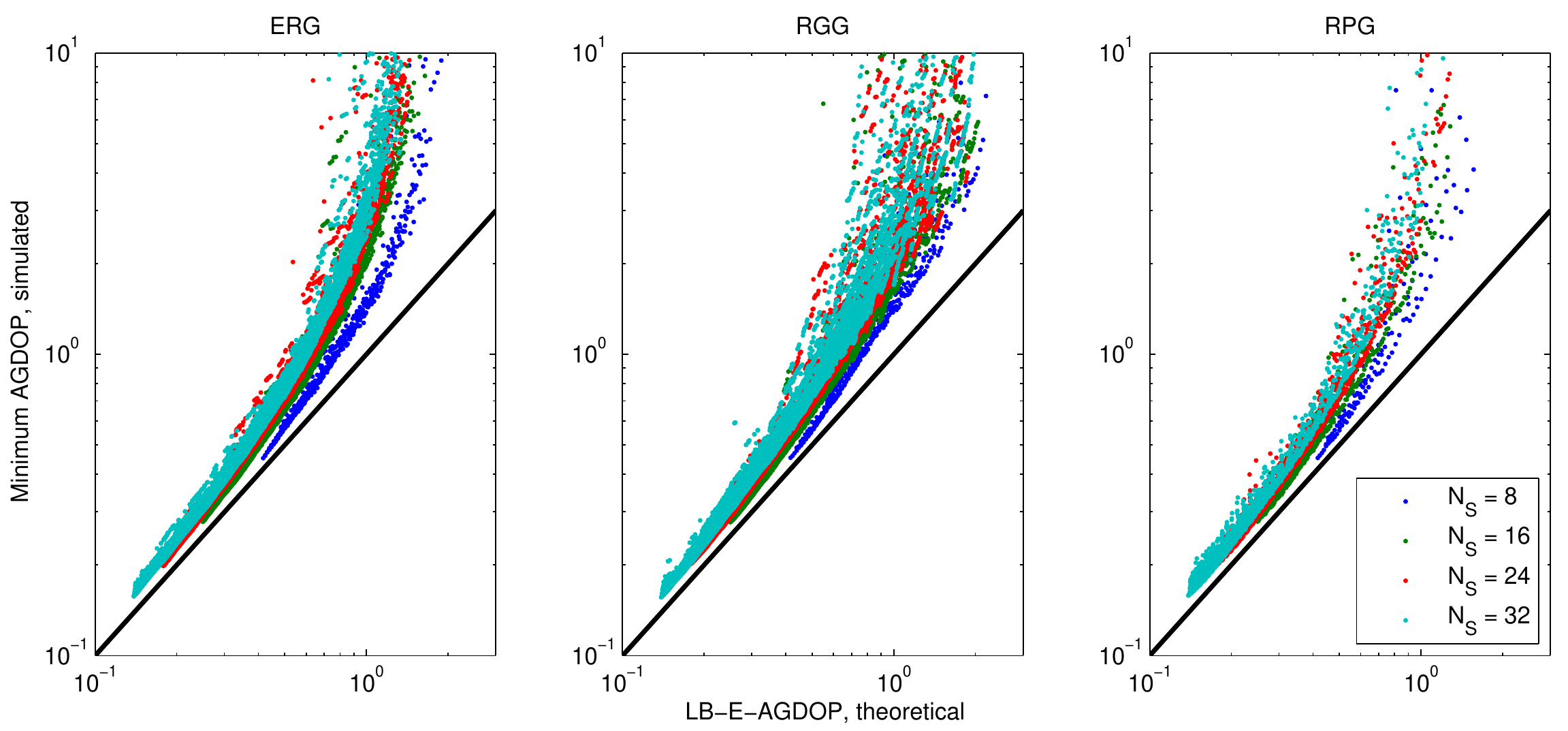}
  \caption{Comparison between LB-E-AGDOP from \eqref{eq:averageGDOP2} and the \textbf{sample minimum} of AGDOP from simulations with the parameters $N_S=8, 16, 24, 32$. The black line shows $y=x$.}\label{fig:DOP_DOP3min}
\end{figure*}

\section{Simulation Results}\label{sec:simResults}

In this section, we conduct numerical simulations to validate the theoretical results obtained in Sections~\ref{sec:LowerBound} and \ref{sec:interpretation}.

\subsection{Simulation settings}

All simulation results presented in this section are based on the following settings.
\begin{itemize}
  \item Two dimensions ($d=2$);
  \item Sensor nodes are uniformly distributed in the unit square $[0,1]\times[0,1]$;
  \item Four anchors ($N_A=4$) located at the corners of the unit square, i.e., $(0,0)$, $(0,1)$, $(1,0)$, and $(1,1)$;
  \item Random graph models: ERG, RGG, and RPG.
\end{itemize}
Fig.~\ref{fig:snrnwk2d} has shown a snapshot excerpted from the simulation of the RGG model with the parameters $r=0.495$ and $N_S=16$.


\subsection{Comparison between the LB-E-AGDOP and E-AGDOP}\label{ssec:simu_E_LB}


Fig.~\ref{fig:DOP_DOP3} compares the LB-E-AGDOP calculated using \eqref{eq:averageGDOP2} and the sample mean of AGDOP obtained from simulations. The simulations are based on the parameters $N_S=8, 16, 24, 32$ and the three random graph models. Each marker in Fig.~\ref{fig:DOP_DOP3} represents a network configuration with certain $N_S$, $\delta_S$ and $\delta_A$. 

Our theoretical lower bound is validated by the simulation results as no markers are below the black line $y=x$. Although the relationship between LB-E-AGDOP and E-AGDOP is nonlinear, if two different network configurations result in the same LB-E-AGDOP value, they also lead to very close E-AGDOP values. Therefore, our derived lower bound, LB-E-AGDOP, can be used as a performance indicator of the expected accuracy of range-based cooperative localization in random sensor networks.

Furthermore, it can be seen that our lower bound is validated for all the three random graph models, with resulting similar gaps between LB-E-AGDOP and E-AGDOP. This demonstrates that our lower bound is applicable to various random graph models, as long as the coordinate symmetry assumption holds.

Fig.~\ref{fig:DOP_DOP3} also shows that the lower bound is tighter when LB-E-AGDOP is smaller. When LB-E-AGDOP is greater than 1, E-AGDOP grows dramatically. In the last part of this section, we shall see that this is because E-AGDOP is likely to be infinite when $\delta<4$. Therefore, if E-AGDOP is finite, LB-E-AGDOP is usually less than 1, and the lower bound is considerably tight.

\subsection{Comparison between the LB-E-AGDOP and minimum AGDOP}\label{ssec:simu_M_LB}

Fig.~\ref{fig:DOP_DOP3min} compares the LB-E-AGDOP calculated using \eqref{eq:averageGDOP2} and the sample minimum of AGDOP obtained from simulations. The simulations are based on the parameters $N_S=8, 16, 24, 32$ and the three random graph models. Each marker in Fig.~\ref{fig:DOP_DOP3} represents a network configuration with certain $N_S$, $\delta_S$ and $\delta_A$. Comparing Fig.~\ref{fig:DOP_DOP3min} to Fig.~\ref{fig:DOP_DOP3}, we can see that the theoretical lower bound matches the minimum AGDOP better than matches the E-AGDOP.

Similar to our discussion about Fig.~\ref{fig:DOP_DOP3}, it can be seen that (1) the theoretical lower bound is validated by the simulation results; (2) the lower bound can be used as a performance indicator of the best accuracy of range-based cooperative localization in random sensor networks; (3) the lower bound works for all the three random graph models as well as other models that satisfy the conditions in Section~\ref{ssec:assumptions}; and (4) the lower bound is tighter when LB-E-AGDOP is smaller.

\subsection{Additional discussion}

When we compare LB-E-AGDOP to E-AGDOP, there is an implicit assumption E-AGDOP${}<\infty$. To many people's surprise, localizability, i.e., uniqueness of the solutions to the localization problem (as treated in \cite{Aspnes06,So2007,Teymorian5291229}), does not necessarily guarantee E-AGDOP${}<\infty$. More specifically, for one sensor node, three anchors almost surely achieve unique localizability. However, by \eqref{eq:trGG0} it can be verified that
$
\E\del[1]{\tr[(G\T G)^{-1}]}<\infty
$
if and only if $N_A\geq 4$. Our ongoing work \cite{HengUP13} has proven that for $d$ dimensions, at least $d+2$ anchors are required to locate one sensor node with finite accuracy.

In cooperative localization, if $\delta<d+2$, there must be one node that has $d+1$ neighbors or fewer. Thus, expectation of GDOP of this node is infinite, so is the E-AGDOP. When $N_S$ is large, LB-E-AGDOP${}\rightarrow d^2/\delta$. Therefore, E-AGDOP is likely to be infinite if
\begin{equation}
\text{LB-E-AGDOP}>d^2/(d+2).
\end{equation}
This explains why in Fig.~\ref{fig:DOP_DOP3}, E-AGDOP grows dramatically when LB-E-AGDOP is greater than 1. This also suggests that in practice, the network connectivity should meet the requirement $\text{LB-E-AGDOP}\leq d^2/(d+2)$ so that the nodes can be accurately located. Furthermore, if this requirement is met, from the simulation results we can see that E-AGDOP and minimum AGDOP are very close, and both of them can be approximated by our lower bound.


\section{Conclusion}\label{sec:conclusion}

This paper has presented a generalized theory that characterizes the connection between system parameters (network connectivity and size) and the accuracy of range-based localization schemes in random WSNs. We have proven a novel lower bound on expectation of AGDOP and derived a closed-form formula \eqref{eq:averageGDOP2} that relates LB-E-AGDOP and E-AGDOP to only three parameters: average sensor degree $\delta_S$, average anchor degree $\delta_A$, and number of sensor nodes $N$. The formula shows that LB-E-AGDOP is approximately inversely proportional to the average degree, and a higher ratio of average anchor degree to average sensor degree leads to better localization accuracy.

The simulation results have validated the theoretical results, and shown that (1) the lower bound are applicable to various random graph models that satisfy our coordinate symmetry assumption; (2) E-AGDOP and minimum AGDOP are very close, and both of them can be approximated by LB-E-AGDOP when LB-E-AGDOP is small. The theory and simulation results presented in this paper provide guidelines on the design of range-based localization schemes and the deployment of sensor networks.

\appendices
\section{Proof of Theorem~\ref{thm:H_tildeH}}\label{ap:1}

There are a few approaches to proving Theorem~\ref{thm:H_tildeH}. One of the simplest proofs is based on a recent result about the Cauchy--Schwarz inequality for the expectation of random matrices \cite{Tripathi19991,lavergne2008cauchy}:
\begin{lemma}[Cauchy--Schwarz inequality \cite{Tripathi19991,lavergne2008cauchy}]\label{lem:matrixCauchy}
Let $A\in\mathbb{R}^{n\times p}$ and $B\in\mathbb{R}^{n\times p}$ be random matrices such that $\E \|A\|^2<\infty$, $\E \|B\|^2<\infty$, and $\E(A\T A)$ is non-singular. Then
\begin{equation}
\E(B\T B)\succeq\E(B\T A)[\E (A\T A)]^{-1}\E(A\T B).
\end{equation}
\end{lemma}

With the substitutions $A=G$ and $B=G(G\T G)^{-1}$ into the above inequality, we have
\begin{equation}
U=\E[(G\T G)^{-1}]\succeq V=[\E (G\T G)]^{-1},
\end{equation}
which already proves Theorem~\ref{thm:H_tildeH}.

Since the diagonal elements of a positive semidefinite matrix must be non-negative,
we have
\begin{equation}\label{eq:XXX}
  U_{ii}\geq V_{ii},\quad\forall i=1,\ldots,dN_S,
\end{equation}
where $U=[U_{ij}]$ and $V=[V_{ij}]$.
In particular, the expectation of GDOP, $\tr(U)$, has a lower bound $\tr(V)$.

\section{Proof of Eq.~(\ref{eq:checkF})}\label{ap:2}

\begin{lemma}[Sherman--Morrison formula \cite{Hager1989}]\label{lem:SMformula}
Suppose $A$ is an invertible square matrix, and $u$ and $v$ are vectors. Suppose furthermore that $1 + v\T A^{-1}u \neq 0$. Then the Sherman--Morrison formula states that
\begin{equation}
    (A+uv\T)^{-1} = A^{-1} - \frac{A^{-1}uv\T A^{-1}}{1 + v\T A^{-1}u}.
\end{equation}
\end{lemma}

With $\eta=d^{-1}[\delta+\delta_S/(N_S-1)]$, \eqref{eq:EFij} can be written as
\begin{equation}
\eta^{-1}\E\check F=I-uu\T,
\end{equation}
where $u=\sqrt{\zeta}(1,1,\ldots,1)\T$, and $\zeta=\delta_S/[\delta(N_S-1)+\delta_S]$.

Letting $u=-v=\sqrt{\zeta}(1,1,\ldots,1)\T$, by the Sherman--Morrison formula we have
\begin{equation}
(I-uu\T)^{-1}=I+uu\T/(1-u\T u),
\end{equation}
and thus
\begin{equation}
\begin{split}
\eta\tr\bigl[(\E \check F)^{-1}\bigr]&=\tr[(I-uu\T)^{-1}]\\
&=N_S+N_S\zeta/(1-N_S\zeta).
\end{split}
\end{equation}

%
%

\ifCLASSOPTIONcaptionsoff
  \newpage
\fi



\bibliographystyle{IEEEtran}
\bibliography{IEEEabrv,../../CoPos}

\begin{IEEEbiography}[{\includegraphics[width=1in,height=1.25in,clip,keepaspectratio]{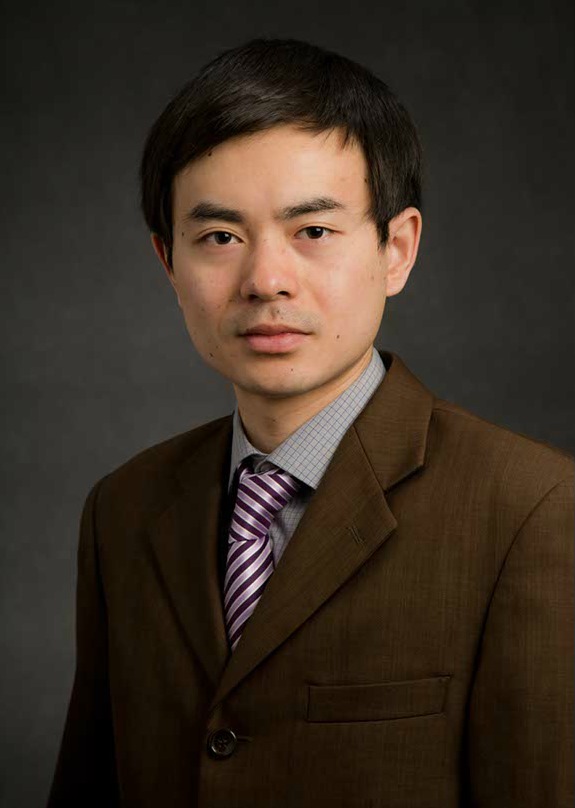}}]{Liang Heng} received the B.S. and M.S. degrees in electrical engineering from Tsinghua University, Beijing, China in 2006 and 2008. He received the PhD degree in electrical engineering from Stanford University under the direction of Per Enge in 2012. He is currently a postdoctoral research associate in the Department of Aerospace Engineering, University of Illinois at Urbana-Champaign. His research interests are cooperative navigation and satellite navigation. He is a member of the IEEE and the Institute of Navigation (ION).
\end{IEEEbiography}

\begin{IEEEbiography}[{\includegraphics[width=1in,height=1.25in,clip,keepaspectratio]{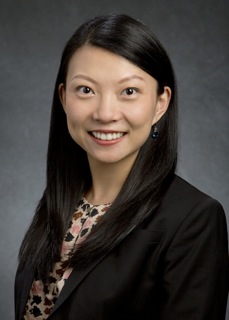}}]{Grace Xingxin Gao} received the B.S. degree in mechanical engineering and the M.S. degree in electrical engineering from Tsinghua University, Beijing, China in 2001 and 2003. She received the PhD degree in electrical engineering from Stanford University in 2008. From 2008 to 2012, she was a research associate at Stanford University. Since 2012, she has been with University of Illinois at Urbana-Champaign, where she is presently an assistant professor in the Aerospace Engineering Department. Her research interests are systems, signals, control, and robotics. She is a member of the IEEE and the Institute of Navigation (ION).
\end{IEEEbiography}


\vfill


\end{document}